\documentclass[aps,pra,twocolumn,superscriptaddress,floatfix,
nofootinbib,showpacs,longbibliography]{revtex4-1}

\usepackage[utf8]{inputenc}  
\usepackage[T1]{fontenc}     
\usepackage[british]{babel}  
\usepackage[sc,osf]{mathpazo}
\usepackage[scaled=0.86]{berasans}  
\usepackage[colorlinks=true, citecolor=blue, urlcolor=blue]{hyperref}  
\usepackage{graphicx} 
\usepackage[babel]{microtype}  
\usepackage{amsmath,amssymb,amsthm,bm,amsfonts,mathrsfs,bbm} 

\usepackage{xspace}  
\usepackage{pgf,tikz}
\usepackage{xcolor}
\usepackage{multirow}
\usepackage{array}
\usepackage{bigstrut}
\usepackage{braket}
\usepackage{color}
\usepackage{natbib}
\usepackage{multirow}
\usepackage{mathtools}
\usepackage{float}
\usepackage[caption = false]{subfig}
\usepackage{xcolor,colortbl}
\usepackage{color}

\newcommand{\be}{\begin{equation}}
\newcommand{\ee}{\end{equation}}
\newcommand{\ba}{\begin{eqnarray}}
\newcommand{\ea}{\end{eqnarray}}

\newtheorem{theorem}{Theorem}

\newtheorem{definition}{Definition}

\newtheorem{observation}{Observation}

\newtheorem{lemma}{Lemma}
\newcommand{\tri}[3]{\ket{#1}\otimes\ket{#2}\otimes\ket{#3}}

\newcommand{\R}{\mathbb{R}}

\def\>{\rangle}
\def\<{\langle}

\begin{document}

\title{Nonlocality Without Entanglement: Quantum Theory and Beyond}

\author{Some Sankar Bhattacharya}
\affiliation{Department of Computer Science, The University of Hong Kong, Pokfulam Road, Hong Kong.}

\author{Sutapa Saha}
\affiliation{Physics and Applied Mathematics Unit, Indian Statistical Institute, 203 B.T. Road, Kolkata-700108, India.}

\author{Tamal Guha}
\affiliation{Physics and Applied Mathematics Unit, Indian Statistical Institute, 203 B.T. Road, Kolkata-700108, India.}

\author{Manik Banik}
\affiliation{S.N. Bose National Center for Basic Sciences, Block JD, Sector III, Salt Lake, Kolkata 700098, India.}

\begin{abstract}
{\it Quantum nonlocality without entanglement} (Q-NWE) encapsulates nonlocal behavior of multipartite product states as they may entail global operation for optimal decoding of the classical information encoded within. Here we show that the phenomena of NWE is not specific to quantum theory only, rather a class of generalized probabilistic theories can exhibit such behavior. In fact several manifestations of NWE,  e.g., asymmetric local discrimination, suboptimal local discrimination, notion of separable but locally unimplementable measurement arise generically in operational theories other than quantum theory. We propose a framework to compare the strength of NWE in different theories and show that such behavior in quantum theory is limited, suggesting a specific topological feature of quantum theory, namely, the continuity of state space structure. Our work adds profound foundational appeal to the study of NWE phenomena along with its information theoretic relevance.        
\end{abstract}



\maketitle
\emph{Introduction.--} 
One of the most counter-intuitive aspects of quantum theory is its {\it nonlocal} behavior. John Bell in his seminal result showed that entangled states of composite quantum systems can result in correlations that do not allow any {\it local realistic} explanation \cite{Bell} (see \cite{Mermin93,Brunner14} for reviews on Bell nonlocality). Such correlations, however, are not available in its most strengthened form \cite{Cirelson80}. This limited behavior of Bell nonlocality has later been axiomatized in deriving quantum theory \cite{Popescu94} and subsequently motivates several computational and information theoretic primitives that aim to single out the correlations allowed in physical world \cite{vanDam,Brassard06,Linden07,Pawlowski09,Navascues09,Fritz13}.

Nonlocal behaviors of quantum theory, however, not always necessitate entanglement. In a pioneering work Bennett {\it et al.} recognized that multipartite quantum systems can exhibit nonlocal properties involving only product states in a way fundamentally different from Bell nonlocality \cite{Bennett99}. In particular, they constructed sets of product states that cannot be exactly discriminated using local operations and classical communication (LOCC) while mutual orthogonality assures their perfect global discrimination. The authors coined the term `quantum nonlocality without entanglement' (Q-NWE) for this phenomenon as the sets allow local preparation (with some preshared strategy) but prohibit perfect local discrimination. Putting differently, global measurement can be more efficient than LOCC for extracting classical information encoded in locally prepared ensemble of quantum states. Local indistinguishability has later been identified as crucial primitive for number of distributed quantum protocols, namely, quantum data hiding \cite{Terhal01,Eggeling02,Matthews09} and quantum secret sharing \cite{Markham08,Rahaman15,Wang17-0}; and it underlies the non-zero gap between single-shot and multi-shot classical capacities of noisy quantum channels \cite{Fuchs96}. On the foundational part, the recent Pusey-Barrett-Rudolph theorem uses such a nonlocal feature of nonorthogonal product states to establish {\it $\psi$-ontic} nature of quantum wave function \cite{Pusey12}.

Here we ask the question what quantum feature is indeed captured in `quantum nonlocality without entanglement'? More particularly we look for whether this phenomenon is specific to quantum theory or is it possible to devise generalized probabilistic models other than quantum mechanics that also manifest similar behavior. Quite surprisingly, like the Bell nonlocality case, here also we find affirmative answer. Recall that indistinguishability of pure states in quantum theory can be thought of as an artifact of Hilbert space structure as it entails nonorthogonality among the states; and the impossibility of perfect local discrimination of an orthogonal product basis, there, results in a {\it separable measurement} that cannot be implemented locally. The generalized probability theory (GPT) framework admits a broader mathematical set-up that includes quantum and classical theory as special examples. It can incorporate the notion of indistinguishable pure states without invoking the much constrained Hilbert space structure of quantum mechanics \cite{Banik19}. Here we show that this general framework can also exhibit nonzero gap in optimal success probabilities of product states discrimination under local and global protocols, {\it i.e.}, it evinces the NWE phenomena. In fact, we find that several aspects of NWE, as observed in quantum theory, are indeed feasible in this generalized framework. For instance, it is possible to have a set of bipartite product states that allows perfect local discrimination when one of the parties starts the protocol whereas the protocol fails for the other party, a fact already known in quantum theory \cite{Groisman01,Walgate02}. In the GPT framework, we then constructed sets of product states that cannot be perfectly discriminated by any local protocol whereas a global measurement serves the purpose exactly. Furthermore, all the effects constituting the perfect discrimination measurement are product effects ensuring the existence of separable but locally unimplementable measurement in GPT framework. This mimics the NWE phenomena as established in the seminal work of Bennett {\it et al.} \cite{Bennett99}. We propose a methodology to compare the strength of NWE in different theories and subsequently show that such behavior in quantum theory is limited in nature. Importantly, it turns out that this limited behavior of NWE in quantum theory is deeply linked with one of its topological feature. More particularly the limited NWE is caused due to the continuity of state space structure in quantum theory which assures existence of a continuous reversible transformation on a system between any two pure states \cite{Hardy01}. Our present work thus adds a deep foundational appeal to the study of NWE phenomena. In the following we start with a brief discussion on GPT framework.  

{\it Generalized probabilistic theory.--} This mathematical framework is broad enough to encapsulate all possible probabilistic theories that use the notion of states to yield the outcome probabilities of measurements (see the Appendix, and Refs. \cite{Hardy01,Barrett07,Chiribella10,Barnum11} for details of this framework). In a GPT, a system $S_{ys}\equiv(\varOmega, \mathcal{E})$ is associated with a set of states $\varOmega$ and a set of effects $\mathcal{E}$. Typically $\varOmega$ is considered to be a compact convex set embedded in the positive convex cone $V_+$ of some real vector space $V$ while $\mathcal{E}$ is embedded in the cone $V^*_+$ which is dual to the state cone $V_+$. An effect $e\in\mathcal{E}$ corresponds to a linear functional on $\varOmega$ that maps each state $\omega\in\varOmega$ onto a probability $p(e|\omega)$, representing successful filter of the effect $e$ on the state $\omega$. Collection of effects $\{e_i\}_i$ forms a measurement whenever $\sum_ie_i=u$, with $u$ being the unit effect such that $p(u|\omega)=1,~\forall~\omega\in\varOmega$. A preparation or state $\omega$ thus specifies outcome probabilities for all measurements that can be performed on it. Given two systems $S_{ys}(A)\equiv(\varOmega^A, \mathcal{E}^A)$ and $S_{ys}(B)\equiv(\varOmega^B, \mathcal{E}^B)$ the theory also delineates their composition. The composite system $S_{ys}(AB)\equiv(\varOmega^{AB}, \mathcal{E}^{AB})$ satisfies some natural conditions, such as no-signaling and local tomography \cite{Hardy01}, that narrow down the possibilities of such compositions and assure that $\varOmega^{AB}$ is embedded in the positive cone of $V^A\otimes V^B$ \cite{Namioka69,Wilce92,Barnum07}. 

Hilbert space description of quantum theory lies within this framework. State of the system is represented by a density operator $\rho\in\mathcal{D}(\mathcal{H})$ \cite{self2} while effects correspond to positive semi-definite operators on $\mathcal{H}$, where $\mathcal{H}$ is the Hilbert space associated with the system. The outcome probabilities follow the Born (trace) rule. The Hilbert space of a composite system consisting the subsystems $A_i$'s is given by $\bigotimes_i\mathcal{H}_{A_i}$, where $\mathcal{H}_{A_i}$ be the Hilbert space of $i^{th}$ subsystems.

Here we recall another class of GPTs, namely polygonal models $\mathcal{P}_{ly}(n)\equiv(\Omega(n),\mathcal{E}(n))$ \cite{Janotta11}. The state spaces $\Omega(n)$ for elementary systems are regular polygons with $n$ vertices. The states and effects are represented by vectors in $\mathbb{R}^3$ and $p(e|\omega)$ is given by usual Euclidean inner product. For a fixed $n$, $\Omega(n)$ is the convex hull of $n$ pure states $\{\omega_i\}_{i=0}^{n-1}$ with $\omega_i:=\left(r_n \cos(2 \pi i/n), r_n \sin(2 \pi i/n),1 \right)^T\in\mathbb{R}^3$; $T$ denotes transpose and $r_n:= \sqrt{\sec(\pi/n)}$. The unit effect is given by $u:=(0,0,1)^T$. The set $\mathcal{E}(n)$ of all possible measurement effects consists of convex hull of zero effect, unit effect, and the extremal effects $\{e_i,\bar{e}_i\}_{i=0}^{n-1}$, where $e_i:= \frac{1}{2}\left(r_n \cos((2 i+1) \pi/n),r_n \sin((2 i+1) \pi/n),1\right)^T$ for even $n$ and $e_i:=\frac{1}{1 + {r_n}^2}\left(r_n \cos(2 \pi i/n),r_n \sin(2 \pi i/n),1\right)^T$ for odd $n$ and $\bar{e}:=u-e$. This class of models has attracted considerable interest in the recent past \cite{Janotta13,Massar14,Safi15,Bhattacharya18}.   

A composite system allows the possibility of a state $\omega_{AB}\in\varOmega^{AB}$ that cannot be prepared as a statistical mixture of some product states, {\it i.e.} $\omega_{AB}\neq\sum_ip_i\omega_A^i\otimes\omega_B^i$ with $\{p_i\}_i$ being a probability distribution. Such states are called entangled states. Entangled effects are defined similarly. Whenever such entangled states and/or entangled effects are invoked in a GPT, they must satisfy the basic self-consistency (SC) condition -- any valid composition of systems, states, effects, and their transformations should produce non-negative conditional probabilities. However, mathematically several self-consistent compositions of the elementary system are possible while considering a multipartite system. For instance, bipartite composition of squre-bit model $\mathcal{P}_{ly}(4)$ allows four such nontrivial compositions -- (i) PR model, (ii) HS Model, (iii) Hybrid model, and (iv) Frozen model \cite{Dall'Arno17}. In quantum case also different, in fact infinite, self-consistent compositions are possible \cite{Namioka69,Wilce92,Barnum07,Arai19}. Among these only the quantum composite state space $\mathcal{D}(\mathcal{H}^{\otimes n})$ possesses the property of {\it self-duality} \cite{Birkhoff36}, which in the GPT framework has recently been derived from a computational primitive \cite{Muller12}. With this prelude we now move to the main part of this work.

{\it Nonlocality without entanglement.--} This particular phenomena is related to multipartite state discrimination problem under local protocols. In the GPT framework the task can be defined as follows. Suppose an $n$-partite state chosen randomly from an ensemble $\{p_i,\omega_{A_1\cdots A_n}^i\}_{i=1}^k$ is distributed among $n$ number of spatially separated parties who know the ensemble but not the exact state and aim to identify it given one copy of the system. However, there is severe restrictions on their action -- they can only perform operations on their respective parts of the composite system and can classically communicate with each other. In general, the ensemble can consist of entangled states, but when studying the NWE phenomena we will consider ensembles of product states only. Before proceeding further let us discuss a bit more about local operations. For that in the following we define separable measurements in GPT framework.
\begin{definition}\label{def1}
	Consider an $n$-partite system $S_{sy}(A_1\cdots A_n)\equiv(\varOmega^{A_1\cdots A_n},\mathcal{E}^{A_1\cdots A_n})$. A measurement $\mathcal{M}\equiv\{e_{A_1\cdots A_n}^i~|~e_{A_1\cdots A_n}^i\in\mathcal{E}^{A_1\cdots A_n}~\forall~i~\&~\sum_ie_{A_1\cdots A_n}^i=u_{A_1\cdots A_n}\}$ is called separable if $e_{A_1\cdots A_n}^i=e_{A_1}^i\otimes\cdots\otimes e_{A_n}^i$ for all $i$, where $e_{X}^i\in\mathcal{E}^X$ with $S_{sy}(X)\equiv(\varOmega^{X},\mathcal{E}^{X})$ being the subsystems for $X\in\{A_1,\cdots, A_n\}$.
\end{definition}
Quite surprisingly all separable quantum measurement are not locally implementable, {\it i.e.}, cannot be realized by LOCC \cite{Bennett99}. Generally a LOCC protocol consists of multi-round steps that make its mathematical characterization hard in quantum theory \cite{Chitambar14}. During such a protocol when some party performs a measurement on her/his subsystem and obtains some outcome then naturally the question arises how the given state gets updated. As noted in Ref. \cite{Chiribella19} any valid post-measurement update rule in a GPT must satisfy some basic consistency requirements imposed through Bayesian reasoning. While the von Neumann-Lüders rule in quantum theory is a consistent update rule, there is no natural way in an arbitrary GPT to come up with such a rule. Despite this we now show that several features of local state discrimination problem, as observed in quantum theory, have similar manifestations in GPT.
\begin{center}
	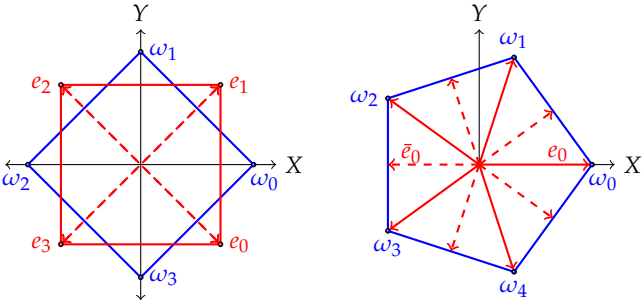
\begin{figure}[!t]
		\begin{tikzpicture}[scale=1.5]
		\newdimen\R
		\R=1cm
		\draw[<->](-1.2,0)--(1.2,0) node[right] {$X$};
		\draw[<->](0,-1.2)--(0,1.2) node[above] {$Y$};
		\draw [blue!100,<->,thick](0:\R) \foreach \x in {0,90,...,270} {-- (\x:\R)} 
		-- cycle (1.1,0) node[below] {$\omega_0$}
		-- cycle (90:\R) node[right] {$\omega_1$}
		-- cycle (-1.1,0) node[below] {$\omega_2$}
		-- cycle (270:\R) node[right] {$\omega_3$};
		\foreach \x in {0,90,...,270} {\draw[ball color=blue] (\x:1 cm) circle (.02cm);}
		\draw [red!100,<->,thick] (45:\R) \foreach \x in {45,135,...,315} {-- (\x:\R)}
		-- cycle (315:\R) node[right] {$e_0$}
		-- cycle (45:\R) node[right] {$e_1$}
		-- cycle (135:\R) node[left] {$e_2$}
		-- cycle (225:\R) node[left] {$e_3$};
		\foreach \x in {45,135,...,315} {\draw[red!100,<->,thick,dashed] (\x:.98 cm) -- (\x + 180:.98cm);
			\draw[ball color=red] (\x:1 cm) circle (.02cm);}
		
		\draw[xshift=3\R][->](1,0)--(1.2,0) node[right] {$X$};
		\draw[xshift=3\R][->](0,0)--(0,1.2) node[above] {$Y$};
		\draw [blue!100,<->,thick] [xshift=3\R](0:\R) \foreach \x in {0,72,...,288} {-- (\x:\R)} 
		-- cycle (1.1,0) node[below] {$\omega_0$}
		-- cycle (72:\R) node[above] {$\omega_1$}
		-- cycle (144:\R) node[left] {$\omega_2$}
		-- cycle (216:\R) node[below] {$\omega_3$}
		-- cycle (288:\R) node[below] {$\omega_4$};
		\foreach \x in {0,72,...,288} {\draw[xshift=3\R] [red!100,<-,thick](\x:.98cm) -- (\x + 180:0cm);
			\draw[xshift=3\R][red!100,->,thick,dashed](\x:0cm) -- (\x + 180:0.807 cm);
			\draw[xshift=3\R][ball color=blue] (\x:1 cm) circle (.02cm);}
		\node[text=red] at (3.7,.1) {$e_0$};
		\node[text=red] at (2.4,.1) {$\bar{e}_0$};
		\end{tikzpicture}
		\caption{[Color on-line] Polygonal models $\mathcal{P}_{ly}(4)$ (left) and $\mathcal{P}_{ly}(5)$ (right). Normalized state plane (at $z=1$) is depicted. Pentagon model is self-dual whereas squit is not. In both models $\omega_i$'s represent pure states. In squit model there are only two extremal measurements, $\mathcal{M}_i\equiv\{e_i,e_{i+1}|e_i+e_{i+1}=u\}$ with $i\in\{0,1\}$, whereas the other has five extremal measurements, $\mathcal{M}_i\equiv\{e_i,\bar{e}_i|e_i+\bar{e}_i=u\}$ with $i\in\{0,\cdots,4\}$. In both models $e_i$'s are ray extremal effects, whereas $\bar{e}_i$'s are extremal effects of $\mathcal{E}(5)$ but not ray extremal.}\label{fig1}
	\end{figure}
\end{center}
\vspace{-.7cm}

We start with the example of asymmetric local discrimination. For composite quantum systems there exist orthogonal product states that can be perfectly discriminated locally if and only if one of the parties starts the protocol \cite{Groisman01,Walgate02}. For instance, the set  $\{\ket{00}_{AB},\ket{01}_{AB},\ket{1+}_{AB},\ket{1-}_{AB}\}\subset\mathbb{C}^2_A\otimes\mathbb{C}^2_B$ is perfectly discriminable when Alice starts the protocol, but not the other way around. Similar is possible in other generalized probabilistic models. 
\begin{lemma}\label{theo1}
Consider the four states, $\$(4):=\{\omega_0\otimes\omega_0,\omega_0\otimes\omega_3,\omega_1\otimes\omega_0,\omega_2\otimes\omega_1\}$ of the composite system $\mathcal{P}_{ly}(4)^{\otimes 2}$, where all $\omega_i\in\Omega(4)$ (see Fig.\ref{fig1}). This set can be perfectly discriminated locally if and only if Alice starts the protocol.
\end{lemma}               
\begin{proof} Squit model $\mathcal{P}_{ly}(4)$ allows two extremal measurements $\mathcal{M}_0\equiv\{e_0,e_2\}$ and $\mathcal{M}_1\equiv\{e_1,e_3\}$ (see Fig.\ref{fig1}). Alice measures $\mathcal{M}_0$ on her part. Outcome corresponding to the effect $e_0$ ensures that the state is one of the first two states of $\$(4)$ that Bob can discriminate through $\mathcal{M}_1$ measurement on his part, otherwise he performs $\mathcal{M}_0$ and perfetly discriminated the other two states. When Bob starts the protocol, whatever measurement he performs the outcomes divide $\$(4)$ into $1$ vs $3$ groups which is impossible for Alice to perfectly discriminate.
\end{proof}
  
As already mentioned, for $\mathcal{P}_{ly}(4)^{\otimes 2}$ four different self-consistent nontrivial compositions are possible along with the trivial minimal composition $\mathcal{P}_{ly}(4)^{\otimes 2}_{min}\equiv\left(\Omega(4)^{\otimes 2}_{min},\mathcal{E}(4)^{\otimes 2}_{min}\right)$, where both the state space $\Omega(4)^{\otimes 2}_{min}$ and the effect space $\mathcal{E}(4)^{\otimes 4}_{min}$ contain only product states and product effects respectively. Among these, the minimal composition and the HS composition \cite{Dall'Arno17} are Bell local models as they contains no entangled state. However, Lemma \ref{theo1} holds true in all of these five models. We now consider a more exotic phenomena of NWE. In their classic paper \cite{Bennett99} Bennett {\it et al.} provided example of orthogonal product bases for $(\mathbb{C}^3)^{\otimes 2}$ and $(\mathbb{C}^2)^{\otimes 3}$ Hilbert spaces that cannot be perfectly discriminated under LOCC operations. However global separable measurements perfectly discriminate the states. To obtain a similar manifestation in GPT we consider the tripartite system $\mathcal{P}_{ly}(5)^{\otimes 3}$. Likewise $\mathcal{P}_{ly}(4)^{\otimes 2}$, here also it will be interesting to find out all possible self-consistent compositions. However, the minimal composition $\mathcal{P}_{ly}(5)^{\otimes 3}_{min}\equiv\left(\Omega(5)^{\otimes 3}_{min},\mathcal{E}(5)^{\otimes 3}_{min}\right)$ will suffice our purpose. 
\begin{theorem}\label{theo2}
Consider the set of states $\$(5)\equiv\{\phi_1:= \omega_{000},~\phi_2:= \omega_{222},~\phi_3:= \omega_{102},~\phi_4:= \omega_{402},~\phi_5:= \omega_{021},~\phi_6:= \omega_{024},~\phi_7:= \omega_{210},~\phi_8:= \omega_{240}\}\subset\Omega(5)^{\otimes 3}_{min}$, where $\omega_{klm}:=\omega_{k}\otimes\omega_{l}\otimes\omega_{m}$ with each $\omega_i\in\Omega(5)$. This set is perfectly discriminable under global operation whereas no local protocol can discriminate them exactly. 
\end{theorem}
\begin{proof}
Symmetry of the states in $\$(5)$ assures that any of the party ( say Alice) can start the local discrimination protocol. Suppose that Alice performs the extremal measurement $\mathcal{M}_0\equiv\{e_0,\bar{e}_0\}$ from such five possible choices $\{\mathcal{M}_i\}_{i=0}^5$ (see Fig. \ref{fig1}). Since, post-measurement update is not well defined in polygonal models, therefore a measurement outcome should either exactly identify the given state or it should conclusively eliminate some possibilities. Alice's outcome $e_0$ ensures that the given state is none of $\{\phi_2,\phi_7,\phi_8\}$. Similarly the outcome $\bar{e}_0$ excludes $\{\phi_1,\phi_5,\phi_6\}$. Having this information Bob and/or Carlie performs suitable measurements on their respective subsystems. At this step one of the outcomes corresponding to the best chosen measurement conclusively eliminates two states. Thus the last party has to identify the state from remaining three states, which is impossible (a flow chart of the protocol is provided in the Appendix). If Alice measures $\mathcal{M}_1\equiv\{e_1,\bar{e}_1\}$, her outcome $e_1$ ($\bar{e}_1$) eliminates only one state $\phi_4$ ($\phi_3$) making the protocol less effective. It is not hard to see that whichever measurement Alice starts with no perfect discrimination is possible.
	
For perfect discrimination, consider the set of effects $\{E_1:= e_{000},~ E_2:= \bar{e}_{000},~E_3:=e_{10},~ \otimes \bar{e}_0,~E_4 := \bar{e}_1 \otimes e_0 \otimes \bar{e}_0,~E_5 =e_0 \otimes \bar{e}_0 \otimes e_1,~E_6 := e_0 \otimes \bar{e}_{01},~E_7 := \bar{e}_0 \otimes e_{10},~E_8: = \bar{e}_{01} \otimes e_0\}\subset\mathcal{E}(5)^{\otimes 3}_{min}$, where $e_{ijk}:=e_i\otimes e_j\otimes e_k$ and $\bar{e}_{lmn}:=\bar{e}_l\otimes \bar{e}_m\otimes \bar{e}_n$, with each $e_r,\bar{e}_s\in\mathcal{E}(5)$. They form a measurement as $\sum_{i=1}^8E_i=u(5)^{\otimes3}$. Straightforward calculation also yields $p(E_i|\phi_j)=\delta_{ij}$ ensuring perfect discrimination of the set $\$(5)$. This concludes the proof. 
\end{proof}
The state discriminable measurement in Theorem \ref{theo2} is a separable measurement (see Definition \ref{def1}). Since no local protocol can perfectly identify the states, therefore this particular separable measurement cannot be implemented locally. Similar construction is also possible in other higher gonal models. Please see the Appendix for explicit constructions in hexagonal and heptagonal models. Regarding the squit model, however, we have the following observation.
\begin{observation}
Construction of nonlocal product states as of Theorem \ref{theo1} is not possible in $\mathcal{P}_{ly}(4)^{\otimes 3}$.
\end{observation}
This follows from the fact that an elementary system of squit does not allow any pair of indistinguishable pure states. In other words, local indistinguishability among pure states is necessary for existence of nonlocal product states of Theorem \ref{theo1}.    

So far we have studied different aspects of NWE in the broader mathematical framework of GPTs. Naturally the question arises how to compare the strength of NWE in different GPTs. To do so, first note that the elementary systems considered in different GPTs must be of {\it same type}. Recall that the phenomena of NWE fundamentally demonstrates difference between local and global operations in extracting classical information encoded in product states. Thus to be in similar footing, different such systems must have same `classical information storage capacity', a notion recently studied for quantum system in Ref. \cite{Weiner15} and generalized for GPTs in Ref. \cite{Dall'Arno17} by the name `signaling dimension'. The concept can be understood with the following communication scenario. Given two finite alphabets $\mathcal{X}=\{x\}$ and $\mathcal{Y}=\{y\}$ of cardinalities $m$ and $n$ respectively, let $\mathcal{P}^{m\rightarrow n}_{S_{ys}}$ denotes the convex set of all $m$-input/$n$-output conditional probability distributions generated by transmitting an elementary system $S_{ys}$ from a sender to a receiver who may have pre-shared randomness. In such a scenario signaling dimension is defined as follows.  
\begin{definition}
(Dall'Arno {\it et al.}, \href{https://doi.org/10.1103/PhysRevLett.119.020401}{PRL {\bf 119}, 020401}) The signaling dimension of a system $S_{ys}$, denoted by $\kappa(S_{ys})$, is defined as the smallest integer d such that $\mathcal{P}^{m\rightarrow n}_{S_{ys}}\subseteq\mathcal{P}^{m\rightarrow n}_{\mathcal{C}_{d}}$, for all	$m$ and $n$. 
\end{definition}
Here, $\mathcal{P}^{m\rightarrow n}_{\mathcal{C}_{d}}$ denotes the set of all $m$-input/$n$-output conditional probability distributions obtained by means of a $d$-dimensional classical noiseless channel and shared random data. Suppose that $\mathcal{S}^1_{ys}$ and $\mathcal{S}^2_{ys}$ are elementary systems of two different theories having identical signaling dimension and $(\mathcal{S}^i_{ys})^{\otimes n}$ be their $n$-partite self-consistent compositions, with $i\in\{1,2\}$. Consider now sets of product states (having same cardinality) in both theses theories that cannot be distinguished locally while respective global measurements can discriminate the states perfectly. The quantity $\Delta[i]:=1-P_L[i]$ amounts to the gap between global and local success probabilities in discriminating the states, where $P_L[i]$ is optimized under all local protocols allowed in the $i^{th}$ theory. If it turns out that $\Delta[1]>\Delta[2]$, then we can say Theory-1 depicts stronger NWE in comparison to Theory-2. In quantum theory both for the systems $(\mathbb{C}^3)^{\otimes 2}$ and $(\mathbb{C}^2)^{\otimes 3}$ we have two different sets of product states with cardinality $8$ that exhibit NWE phenomena \cite{Bennett99}. However, the example of $(\mathbb{C}^2)^{\otimes 3}$ is comparable with that of Theorem \ref{theo2}, since both the elementary systems have same signaling dimension and both examples consider tripartite composite systems. What follows next is the comparison between the strength of NWE in these two examples.  
\begin{theorem}\label{theo3}
$\Delta[\mbox{QT}]\le\frac{1}{8}(4-\sqrt{10})<\frac{1}{8}=\Delta[\mbox{Pentagon}]$.
\end{theorem}
Proof of the theorem is provided in Appendix. Here it is worth mentioning that for pentagon model local success probability is optimized over all possible local protocols which consists of $1$-way protocols only. The corresponding quantum value is also obtained under $1$-way LOCC protocol. More general local protocol (consisting $2$-way LOCC) may further decrease the value of $\Delta[\mbox{QT}]$ which in general is difficult to optimize \cite{Croke17}. From the structure of constructions arguably it follows that other polygonal models also have $\Delta=1/8$. Furthermore, we note that instead of the uniform priori distribution of the states $\{\phi_i\}_{i=1}^8$ if one consider a biased prior distribution (a more feasible situation for the experimental purpose) then also limited NWE behavior of quantum theory can be established (see the Appendix). Importantly, the continuity of the quantum state space plays the crucial role in resulting limited NWE behavior compared to the discrete polygonal models.

One may ask for the GPT analog of the nonlocal product bases in $(\mathbb{C}^3)^{\otimes 2}$ \cite{Bennett99}. We have a negative impression at this point that such analogy will not be possible in bipartite composition of the polygonal models. This is due to the fact that all the polygonal models have signaling dimension $2$ whereas that of the qutrit system is $3$. Of course a rigorous mathematical proof of this intuition will be worth interesting. Such a result will generalize the Theorem $4$ of Ref. \cite{Walgate02} in GPT framework.        
    
{\it Discussion.--} In this work we study the `nonlocality without entanglement' phenomena in the broader mathematical framework of generalized probabilistic theories. We show that this particular nonclassical behavior in quantum theory is limited as compared to other GPT models. This, in a sense, is similar to the fact of limited Bell nonlocality in quantum theory as observed by Rohrlich and Popescu in their seminal work \cite{Popescu94}. In fact, Rohrlich-Popescu proposed to axiomatized this limited Bell nonlocal behavior (along with relativistic causality) to derive quantum theory. Subsequently, several information and physical principles have been proposed to explain limited Bell nonlocality in quantum theory \cite{vanDam,Brassard06,Linden07,Pawlowski09,Navascues09,Fritz13}, and its connection with other quantum features have also been established \cite{Oppenheim10,Banik13,Stevens14,Banik15}. In our work, we observe that the limited nonlocality without entanglement feature in quantum theory is owing to the continuity of quantum state space structure which is presumed in axiomatic derivation of quantum theory either directly \cite{Hardy01}, or invoked through other assumptions such as `reversible transformation' \cite{Masanes11} or `purification' \cite{Chiribella11}. However, the present work shows that the same feature can be obtained as an emerging fact if we presume the limited NWE as a fundamental characteristic of the theory. It therefore welcomes novel information/physical primitive(s) to explain this limited NWE behavior in quantum theory. 

Our work also motivates further research to study other exotic features of NWE phenomena in GPT framework. For instance, the phenomena of NWE in quantum theory was first anticipated by Peres and Wootters \cite{Peres91}. They conjectured that LOCC measurements are suboptimal for discrimination of a specific set of nonorthogonal product states (see also \cite{Massar95}). More recently, the authors in Ref.\cite{Chitambar13} revisited the classic problem of Peres and Wootters and proved that there conjecture is indeed true. A similar example in GPT framework is yet to be obtained. On the other hand, multipartite generalization of NWE phenomena have been studied very recently \cite{Halder19,Rout19}. Similar study in the GPT framework may provide further insight about the structure of quantum theory.              

\begin{acknowledgements}
This work is supported by the National Natural Science Foundation of China through grant 11675136, the Foundational Questions Institute through grant FQXi-RFP3-1325, the Hong Kong Research Grant Council through grant 17300918, and the John Templeton Foundation through grant 60609, Quantum Causal Structures. The opinions expressed in this publication are those of the authors and do not necessarily reflect the views of the John Templeton Foundation. MB acknowledges support through an INSPIRE-faculty position at S. N. Bose National Center for Basic Sciences by the Department of Science and Technology, Government of India. We gratefully acknowledge discussions with Guruprasad Kar, Saronath Halder, Amit Mukherjee, Arup Roy, Mir Alimuddin, and Sumit Rout. We would like to thank Amit Mukherjee for his help with the Figure \ref{fig3} and Saronath Halder for pointing out the Reference \cite{Croke17}. We thankfully acknowledge the useful suggestions by Giulio Chiribella.   
\end{acknowledgements}

\onecolumngrid
\appendix
\section{Generalized probabilistic theories}
\subsection{Preliminaries} 
{\bf (A) State space}: State space $\varOmega$ of a system $S_{ys}\equiv(\varOmega, \mathcal{E})$ is a convex-compact set embedded in the positive cone $V_+$ of some real vector space $V$. Convexity assures statistical mixture of different valid preparations as another valid preparation, {\it i.e.}, for $\omega_1,\omega_2\in\varOmega$ any probabilistic mixture $p\omega_1+(1-p)\omega_2\in\varOmega$, where $p\in[0,1]$. $\varOmega$ is considered to be topologically closed, {\it i.e.}, there is no physical distinction between states that can be prepared
exactly and states that can be prepared to arbitrary accuracy. Furthermore finite dimensionality of $V$ guarantees compactness of $\varOmega$. 

{\bf (B) Effect space}: Effects are linear functional on $\Omega$ that maps each state onto a probability. The set of all linear functionals is as $\Omega^*\subset V_+^*$. Framework of GPTs may assume, a priori, that all mathematically well-defined states and observables are not physically implementable. For example, the set of physically allowed effects $\mathcal{E}$ may be a strict subset of $\Omega^*$. A theory in which all elements of $\Omega^*$ are allowed effects is called ‘dual’. The property of duality is often assumed as a starting point in derivations of quantum theory and referred to as the `no-restriction hypothesis' \cite{Chiribella11}. However, recently it has been shown that the set of `almost-quantum correlations' violates the no-restriction hypothesis \cite{Sainz18}. A $d$-outcome measurement $M$ is specified by a collection of $d$ effects, {\it i.e.}, $M\equiv\{\mathit{e}_j~|~\sum_je_j=u\}$. In GPT framework the notion of distinguishability is defined in the following sense.
\begin{definition}
Members of a set of $n$ states $\{\omega_i\}_{i=1}^n\subset\varOmega$ are called distinguishable if they can be perfectly identified in a single shot measurement {\it i.e.}, if there exists an $n$-outcome measurement $M=\{e_j~|~\sum_{j=1}^ne_j=u\}$ such that $p(\mathit{e}_j|w_i)=\delta_{ij}$.
\end{definition}

{\bf (C) transformation}: Transformation maps states into state (also effects into effects), {\it i.e.}, $T:V\mapsto V$, with $T(V_+)\subseteq V_+$. They are linear and preserve statistical mixtures. They cannot increase the total probability, but are allowed to decrease it.

{\bf (D) Joint system}: For subsystems $S_{ys}(A)\equiv(\varOmega^A, \mathcal{E}^A)$ and $S_{ys}(B)\equiv(\varOmega^B, \mathcal{E}^B)$ a GPT also specifies their composition $S_{ys}(AB)\equiv(\varOmega^{AB}, \mathcal{E}^{AB})$. Clearly $\varOmega^{AB}$ is convex by definition. In general, one can imagine many weird and wonderful relations among $\varOmega^{A}$, $\varOmega^{B}$, and $\varOmega^{AB}$. One can, however, narrow down the possibilities by imposing the following quite natural conditions:
\begin{itemize}
	\item[(i)] Ever joint state $\omega_{AB}\in\varOmega^{AB}$ should assign joint probability to pair of effects $(e_A, e_B)$; $e_A\in \mathcal{E}^A$ and $e_B\in \mathcal{E}^B$.
	\item[(ii)] Joint probabilities must respect the no-signaling condition, {\it i.e.}, the marginal outcome probabilities of the measurements on $B$ should not depend on the measurement choice on $A$ and vice versa.
	\item[(iii)] The joint probabilities for the pairs of effects $(e_A, e_B)$ specifies the joint state. This condition is known as \emph{local tomography} condition \cite{Hardy01}. The real-vector-space quantum theory does not satisfy this condition \cite{Hardy12}.
\end{itemize}
Any joint state space $\varOmega^{AB}$ satisfying the aforesaid requirements is embedded in the positive cone of $V_A\otimes V_B$. Furthermore, it must lie between two
extremes, the \emph{maximal} and the \emph{minimal tensor product} \cite{Namioka69,Wilce92,Barnum07}.
\begin{definition}
The maximal tensor product $\varOmega^A\otimes_{max}\varOmega^B$, is the set of all bilinear functionals $\phi:~V^*_A\otimes V^*_B\rightarrow \mathbb{R}$ such that (i) $\phi(e_A,e_B)\ge 0$ for all $e_A\in \mathcal{E}^{A}$ and $e_B\in \mathcal{E}^{B}$ and (ii) $\phi(u_A,u_B)=1$, where $u_A$ and $u_B$ are unit effects for system $A$ and $B$ respectively. The maximal tensor product has an important operational characterization: it is the largest set of states assigning probabilities to all product measurements but not allowing signaling.
\end{definition}
\begin{definition}
The minimal tensor product $\varOmega^A\otimes_{min}\varOmega^B$, is the convex hull of the product states, where a product $\omega_A\otimes\omega_B$ is defined by $p(a,b|\omega_A\otimes\omega_B):=p(a|\omega_A)p(b|\omega_B)$ for all pairs of $(a,b)\in V^*_A\otimes V^*_B$.
\end{definition}
\subsection{Polygon theories}
{\bf (A) Elementary systems:} We denote the system as $\mathcal{P}_{ly}(n)\equiv(\Omega(n),\mathcal{E}(n))$. For a fixed $n$, $\Omega(n)$ is the convex hull of $n$ pure states $\{\omega_i\}_{i=0}^{n-1}$ with $\omega_i:=\left(r_n \cos(2 \pi i/n), r_n \sin(2 \pi i/n),1 \right)^T\in\mathbb{R}^3$; where $T$ denotes transpose and $r_n:= \sqrt{\sec(\pi/n)}$. The unit effect is given by $u:=(0,0,1)^T$. The set $\mathcal{E}(n)$ is the convex hull of zero effect, unit effect, and the extremal effects $\{e_i,\bar{e}_i\}_{i=0}^{n-1}$, where $e_i:= \frac{1}{2}\left(r_n \cos((2 i-1) \pi/n),r_n \sin((2 i-1) \pi/n),1\right)^T$ for even $n$ and $e_i:=\frac{1}{1 + {r_n}^2}\left(r_n \cos(2 \pi i/n),r_n \sin(2 \pi i/n),1\right)^T$ for odd $n$.

The pure effects $\{e_i\}_{i=0}^{n-1}$ correspond to exposed rays and consequently the extreme rays of $V_+^*$ \cite{Yopp07}. For odd-gonal cases, due to self-duality of state cone $V_+$ and its effect cone $V_+^*$ \cite{Kimura14} every pure effect $e_i$ has one to one ray-correspondence to the pure state $\omega_i$. Consequently, for every pure state $\omega_i$ there exist exactly two other pure states $\omega_i^\prime$ and $\omega_i^{\prime\prime}$ such that $\omega_i$ and $\bar{\omega}_i^{(\eta)}:=\eta\omega_i^\prime+(1-\eta)\omega_i^{\prime\prime}$ are always perfectly distinguishable for all $\eta\in[0,1]$. The discriminating measurement consists of the effects $\{e_i,\bar{e}_i\}$ such that $p(e_i|\omega_i)=1$ and $p(e_i|\bar{\omega}_i^{(\eta)})=0$. The effects $\{\bar{e}_i\}_{i=0}^{n-1}$ are extremal elements of $\mathcal{E}(n)$ but they are not ray extremal, i.e., they do not lie on an extremal ray of the cone $V_+^*$. For an even-gon, the scenario is quite different as the self duality between $V_+$ and $V_+^*$ is absent. Here all the $e_i$'s and their complementary effects $\bar{e}_i$'s correspond to extreme rays of $V_+^*$.

\section{Technical results}
\subsection{NWE in other polygonal models}
Construction as of Proposition-3 is also possible in other polygonal models. Explicit constructions are given for hexagonal and heptagonal models.
\begin{center}
	\begin{figure}[h!]	
	\begin{tikzpicture}[scale=2.2]
	\newdimen\R
	\R=1cm
	\draw[->](0,0)--(1.2,0) node[right] {\large $X$};
	\draw[->](0,1)--(0,1.2) node[above] {\large $Y$};
	\draw [blue!100,<->,thick](0:\R) \foreach \x in {0,60,...,300} {-- (\x:\R)}
	-- cycle (1.1,0) node[below] {\large $\omega_0$}
	-- cycle (60:\R) node[right] {\large $\omega_1$}
	-- cycle (120:\R) node[left] {\large $\omega_2$}
	-- cycle (-1.1,0) node[below] {\large $\omega_3$}
	-- cycle (240:\R) node[left] {\large $\omega_4$}
	-- cycle (300:\R) node[right] {\large $\omega_5$};
	\foreach \x in {0,60,...,300} {\draw[ball color=blue] (\x:1 cm) circle (.02cm);}
	\draw [red!100,<->,thick] (30:\R) \foreach \x in {30,90,...,330} {-- (\x:\R)}
	-- cycle (330:\R) node[right] {\large $e_0$}
	-- cycle (30:\R) node[above] {\large $e_1$}
	-- cycle (90:\R) node[right] {\large $e_2$}
	-- cycle (150:\R) node[left] {\large $e_3$}
	-- cycle (210:\R) node[left] {\large $e_4$}
	-- cycle (270:\R) node[below] {\large $e_5$};
	\foreach \x in {30,90,...,330} {\draw[red!100,<->,thick,dashed] (\x:.98 cm) -- (\x + 180:.98 cm);
		\draw[ball color=red] (\x:1 cm) circle (.02cm);}
	
	\draw[xshift=4\R][->](1,0)--(1.2,0) node[right] {\large $X$};
	\draw[xshift=4\R][->](0,0)--(0,1.2) node[above] {\large $Y$};
	\draw[xshift=4\R] [blue!100,<->,thick](0:\R) \foreach \x in {0,51.428,...,308.571} {-- (\x:\R)}
	-- cycle (1.1,0) node[below] {\large $\omega_0$}
	-- cycle (51.428:\R) node[above] {\large $\omega_1$}
	-- cycle (102.857:\R) node[above] {\large $\omega_2$}
	-- cycle (154.286:\R) node[left] {\large $\omega_3$}
	-- cycle (205.714:\R) node[left] {\large $\omega_4$}
	-- cycle (257.143:\R) node[below] {\large $\omega_5$}
	-- cycle (308.571:\R) node[below] {\large $\omega_6$};
	\foreach \x in {0,51.428,...,308.571} {
		\draw [xshift=4\R] [red!100,<-,thick](\x:.98 cm) -- (\x + 180:0 cm);
		\draw [xshift=4\R] [red!100,->,thick,dashed](\x:0 cm) -- (\x + 180:.89 cm);
		\draw [xshift=4\R] [ball color=blue] (\x:1 cm) circle (.02cm);}
	\node[text=red] at (4.8,.1) {\large $e_0$};
	\node[text=red] at (3.3,.1) {\large $\bar{e}_0$};
	\end{tikzpicture}
		\caption{[Color on-line] Polygonal models $\mathcal{P}_{ly}(6)$ (left) and $\mathcal{P}_{ly}(7)$ (right). In $\mathcal{P}_{ly}(6)$, each extremal effect  exactly filters two extremal states perfectly whereas does not filter another two extremal states at all and filters other pure states probabilistically. For instance, for the effects $e_1$,  $p(e_1|\omega_0)=p(e_1|\omega_0)=1$, $p(e_1|\omega_3)=p(e_1|\omega_4)=0$, and $p(e_1|\omega_2)=p(e_1|\omega_5)\in(0,1)$. On the other hand, in heptagon model, each ray-extremal effect perfectly filters only one pure state, and does not filter two pure states at all, and rest pure states are filtered probabilistically. For example, $p(e_0|\omega_0)=1$, $p(e_0|\omega_3)=p(e_0|\omega_4)=0$, and $p(e_0|\omega_i)\in(0,1)$ for $i\in\{1,2,5,6\}$. For the non-ray extremal effects the scenario is bit different, $p(\bar{e}_0|\omega_3)=p(\bar{e}_0|\omega_4)=1$, $p(\bar{e}_0|\omega_0)=0$, and $p(\bar{e}_0|\omega_i)\in(0,1)$ for $i\in\{1,2,5,6\}$.}\label{fig2}
	\end{figure}
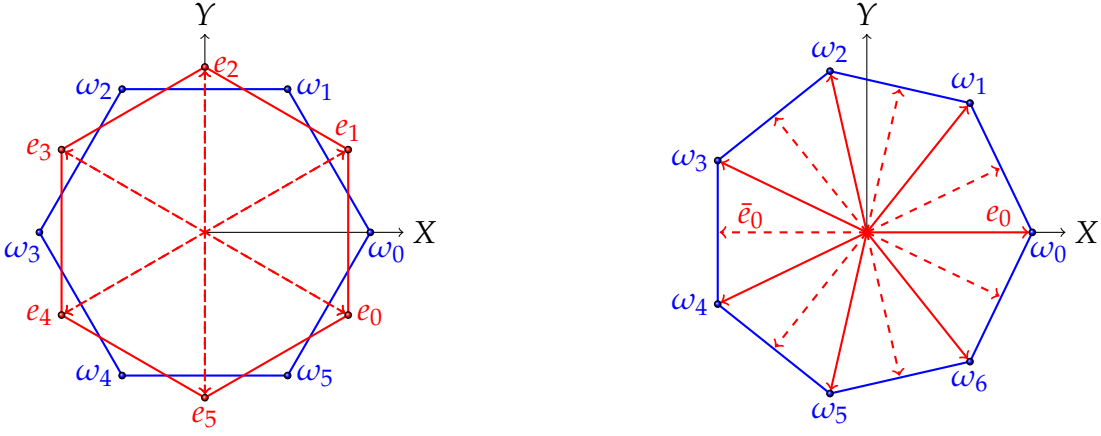
\end{center}

{\bf (a)} Construction in $\mathcal{P}_{ly}(6)^{\otimes 3}_{min}\equiv\left(\Omega(6)^{\otimes 3}_{min},\mathcal{E}(6)^{\otimes 3}_{min}\right)$: the set of states are given by,
\begin{align}\label{hexas}
\rotatebox[origin=c]{0}{$\$(6)\equiv$}
\left\{\!\begin{aligned}
\phi_1:= \omega_{0}\otimes\omega_{0}\otimes\omega_{0},~~~\phi_2:= \omega_{3}\otimes\omega_{3}\otimes\omega_{3},~~~
\phi_3:= \omega_{1}\otimes\omega_{0}\otimes\omega_{3},~~~\phi_4:= \omega_{5}\otimes\omega_{0}\otimes\omega_{3},\\
\phi_5:= \omega_{0}\otimes\omega_{3}\otimes\omega_{1},~~~\phi_6:= \omega_{0}\otimes\omega_{3}\otimes\omega_{5},~~~
\phi_7:= \omega_{3}\otimes\omega_{1}\otimes\omega_{0},~~~\phi_8:= \omega_{3}\otimes\omega_{5}\otimes\omega_{0}~
\end{aligned}\right\},
\end{align}
where each $\omega_i\in\Omega(6)$. Local indistinguishability follows like in Proposition-3. The separable measurement that perfectly discriminate this set is given by,    
\begin{align}\label{hexam}
\rotatebox[origin=c]{0}{}
\left\{\!\begin{aligned}
E_1 = e_0 \otimes e_0 \otimes e_0,~~~E_2 = e_3 \otimes e_3 \otimes e_3,~~~E_3 =e_1 \otimes e_0 \otimes e_3,~~~E_4 = e_4 \otimes e_0 \otimes e_3,\\
E_5 =e_0 \otimes e_3 \otimes e_1,~~~E_6 = e_0 \otimes e_3 \otimes e_4,~~~E_7 = e_3 \otimes e_1 \otimes e_0,~~~E_8 = e_3 \otimes e_4 \otimes e_0~
\end{aligned}\right\},
\end{align}
where each $e_i\in\mathcal{E}(6)$. It is a straightforward observation that $p(E_i|\phi_j)=\delta_{ij}$ (see Fig.\ref{fig2}). 

{\bf (b)} Construction in $\mathcal{P}_{ly}(7)^{\otimes 3}_{min}\equiv\left(\Omega(7)^{\otimes 3}_{min},\mathcal{E}(7)^{\otimes 3}_{min}\right)$: the set of states are given by,
\begin{align}\label{heptas}
\rotatebox[origin=c]{0}{$\$(7)\equiv$}
\left\{\!\begin{aligned}
\phi_1:= \omega_{0}\otimes\omega_{0}\otimes\omega_{0},~~~\phi_2:= \omega_{3}\otimes\omega_{3}\otimes\omega_{3},~~~
\phi_3:= \omega_{1}\otimes\omega_{0}\otimes\omega_{3},~~~\phi_4:= \omega_{5}\otimes\omega_{0}\otimes\omega_{3},\\
\phi_5:= \omega_{0}\otimes\omega_{3}\otimes\omega_{1},~~~\phi_6:= \omega_{0}\otimes\omega_{3}\otimes\omega_{5},~~~
\phi_7:= \omega_{3}\otimes\omega_{1}\otimes\omega_{0},~~~\phi_8:= \omega_{3}\otimes\omega_{5}\otimes\omega_{0}~
\end{aligned}\right\},
\end{align}
where each $\omega_i\in\Omega(7)$. Local indistinguishability follows like in Proposition-3. The separable measurement that perfectly discriminate this set is given by,    
\begin{align}\label{heptam}
\rotatebox[origin=c]{0}{}
\left\{\!\begin{aligned}
E_1 = e_0 \otimes e_0 \otimes e_0,~~~E_2 = \bar{e}_0 \otimes \bar{e}_0 \otimes \bar{e}_0,~~~E_3 =e_1 \otimes e_0 \otimes \bar{e}_0,~~~E_4 = \bar{e}_1 \otimes e_0 \otimes \bar{e}_0,\\
E_5 =e_0 \otimes \bar{e}_0 \otimes e_1,~~~E_6 = e_0 \otimes \bar{e}_0 \otimes \bar{e}_1,~~~E_7 = \bar{e}_0 \otimes e_1 \otimes e_0,~~~E_8 = \bar{e}_0 \otimes \bar{e}_1 \otimes e_0~
\end{aligned}\right\},
\end{align}
where each $e_i,\bar{e}_j\in\mathcal{E}(7)$. It is again straightforward to check that $p(E_i|\phi_j)=\delta_{ij}$ (see Fig.\ref{fig2}).

\subsection{Strength of NWE in different theories}
\begin{figure}[h!]
	\centering
	\includegraphics[width=.8\textwidth]{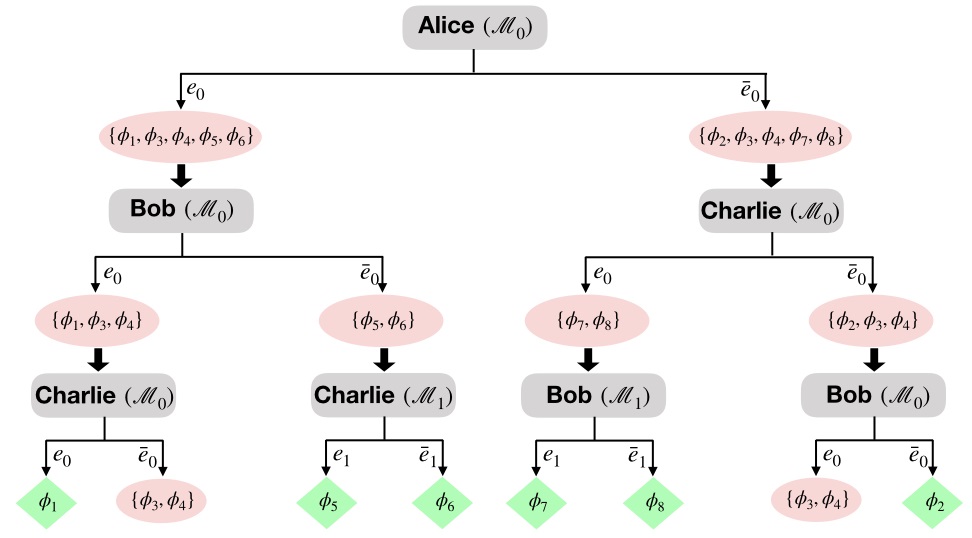}
	\caption{[Color on-line] The optimal local protocol for distinguishing the set $\$(5)$. In the last step, if there is a hit in square shaped box the state is conclusively identified, else ambiguity arises. This flow chart of optimal local discrimination looks similar for the other polygonal models.}\label{fig3}
\end{figure}
{\bf Pentagon model:} We are interested in optimal success probability of local discrimination of the states,
\begin{align}\label{pent}
\rotatebox[origin=c]{0}{$\$(5)\equiv$}
\left\{\!\begin{aligned}
\phi_1:= \omega_{0}\otimes\omega_{0}\otimes\omega_{0},~~~\phi_2:= \omega_{2}\otimes\omega_{2}\otimes\omega_{2},~~~
\phi_3:= \omega_{1}\otimes\omega_{0}\otimes\omega_{2},~~~\phi_4:= \omega_{4}\otimes\omega_{0}\otimes\omega_{2},\\
\phi_5:= \omega_{0}\otimes\omega_{2}\otimes\omega_{1},~~~\phi_6:= \omega_{0}\otimes\omega_{2}\otimes\omega_{4},~~~
\phi_7:= \omega_{2}\otimes\omega_{1}\otimes\omega_{0},~~~\phi_8:= \omega_{2}\otimes\omega_{4}\otimes\omega_{0}~
\end{aligned}\right\}.
\end{align}
As discussed in the proof of Proposition-2, due to party symmetry any party can start the local discrimination protocol. Starting with Alice's measurement $\mathcal{M}_0=\{e_0,\bar{e}_0\}$, the full discrimination protocol is diagrammatically depicted in Fig \ref{fig3}. The above protocol shows that out of $8$ states $6$ can be discriminated perfectly whereas confusion arises for the other $2$ states. Since the states are given with uniform random probability, the the success probability turns out to be,
\begin{eqnarray}
P^{succ}_L=\frac{1}{8}\left( 6\times 1+2\times\frac{1}{2}\right)=\frac{7}{8}. 
\end{eqnarray}
It follows from the proof of Proposition-2, any other strategy is no good for yielding a greater success probability. Therefore we have,
\begin{equation}
\Delta[\mbox{Pentagon}]=1-\frac{7}{8}=\frac{1}{8}. 
\end{equation} 
It is not hard to see that in other polygonal model also we have $\Delta=\frac{1}{8}$.
\begin{center}
	\begin{figure}[b!]
		\begin{tikzpicture}[scale=1.5]
		\draw [thick,blue](0,0) circle (1cm);
		\draw [thick,blue](1,0) node[right]{$\large{\ket{+}}$}--(-1,0)node[left]{$\ket{-}$};
		\draw [thick,blue](0,1)node[above]{$\ket{0}$}--(0,-1)node[below]{$\ket{1}$};
		\draw [thick,red](250:1cm)node[below]{$\ket{\theta^\perp}$}--(70:1cm)node[above]{$\ket{\theta}$};
		\draw[thick,black] (70:.45cm) arc (70:90:13pt);
		\node at (80:.6cm) {$\theta$};
		\end{tikzpicture}
		\caption{[Color on-line] $XZ$ plane of the Bloch sphere. Instead of measuring in computation basis, Alice performs a measurement in $\{\ket{\theta},\ket{\theta^\perp}\}$ basis. Error gets minimized at $\theta=\tan^{-1}\left( \frac{1}{3}\right) $.}\label{fig4}
	\end{figure}
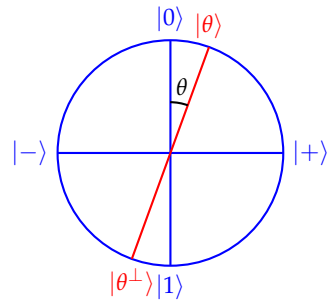
\end{center}

{\bf Quantum theory:} A similar type example in quantum theory is provided by Bennett {\it el al.} in 3-qubit system \cite{Bennett99}. The set of states are given by,
\begin{align}\label{quantum}
\rotatebox[origin=c]{0}{$\mathcal{Q}\equiv$}
\left\{\!\begin{aligned}
\phi_1:= \tri{0}{0}{0},~~~\phi_2:= \tri{1}{1}{1},~~~
\phi_3:= \tri{+}{0}{1},~~~\phi_4:= \tri{-}{0}{1},\\
\phi_5:= \tri{0}{1}{+},~~~\phi_6:= \tri{0}{1}{-},~~~
\phi_7:= \tri{1}{+}{0},~~~\phi_8:= \tri{1}{-}{0}~
\end{aligned}\right\},
\end{align}
where $\{\ket{0},\ket{1}\}$ are the eigenstates of Pauli $\sigma_z$ operator and $\ket{\pm}:=(\ket{0}\pm\ket{1})/\sqrt{2}$. In general it is very hard to characterize the set of LOCC operation in quantum theory \cite{Chitambar14}. It is also difficult to find out the optimal discrimination probability under such protocols \cite{Croke17}. So for discriminating the set $\mathcal{Q}$, we first consider 1-way LOCC protocol, where one of the party starts the protocol and based on her/his result one of the remaining two parties do some local operation and finally the third party performs local operation and try to guess the state. 

Due to symmetry in construction, for the set $\mathcal{Q}$, any of the parties can start the protocol. If Alice starts by performing a measurement in Pauli $\sigma_z$ basis and they follow a protocol like in the pentagon model then the success probability turns out to be $7/8$ (to see it replace $\mathcal{M}_0$ by $\sigma_z$ and $\mathcal{M}_1$ by $\sigma_x$ in the flow chart \ref{fig3}). Instead of this let Alice perform a measurement in $\{\ket{\theta},\ket{\theta^\perp}\}$ basis (see Fig \ref{fig4}). Based on her outcome $\ket{\theta}$ and $\ket{\theta^\perp}$ she divides the states into two groups $\mathcal{G}_1\equiv\{\phi_1,\phi_3,\phi_5,\phi_6\}$ and $\mathcal{G}_2\equiv\{\phi_2,\phi_4,\phi_7,\phi_8\}$, respectively and informs her outcome to Bob and Charlie. Note that both these groups are perfectly local discriminable by Bob and Charlie. Therefore the error occurs at Alice's step only. The error can occur in two ways, the state $\ket{\theta}$ may clicks even when the state is from $\mathcal{G}_2$ and similarly $\ket{\theta^\perp}$ may clicks even when the state is from $\mathcal{G}_1$. The total error is thus given by,
\begin{equation}
P_{err}(\theta)=\frac{1}{8}\times 2\left(3\times\frac{1}{2}(1-\cos\theta) +\frac{1}{2}(1-\sin\theta)\right).
\end{equation}  
Optimizing over $\theta$, we obtain $P_{err}=\frac{1}{8}(4-\sqrt{10})$ at $\theta=\tan^{-1}\left( \frac{1}{3}\right) $, and thus $P^{succ}_L=1-P_{err}$ which subsequently yield,
\begin{equation}
\Delta[\mbox{Quantum}]\le\frac{1}{8}(4-\sqrt{10})<\frac{1}{8}=\Delta[\mbox{Pentagon}].
\end{equation} 
\subsection{Ensemble with biased probability distribution}
\begin{figure}[b!]
	\centering
	\includegraphics[width=.5\textwidth]{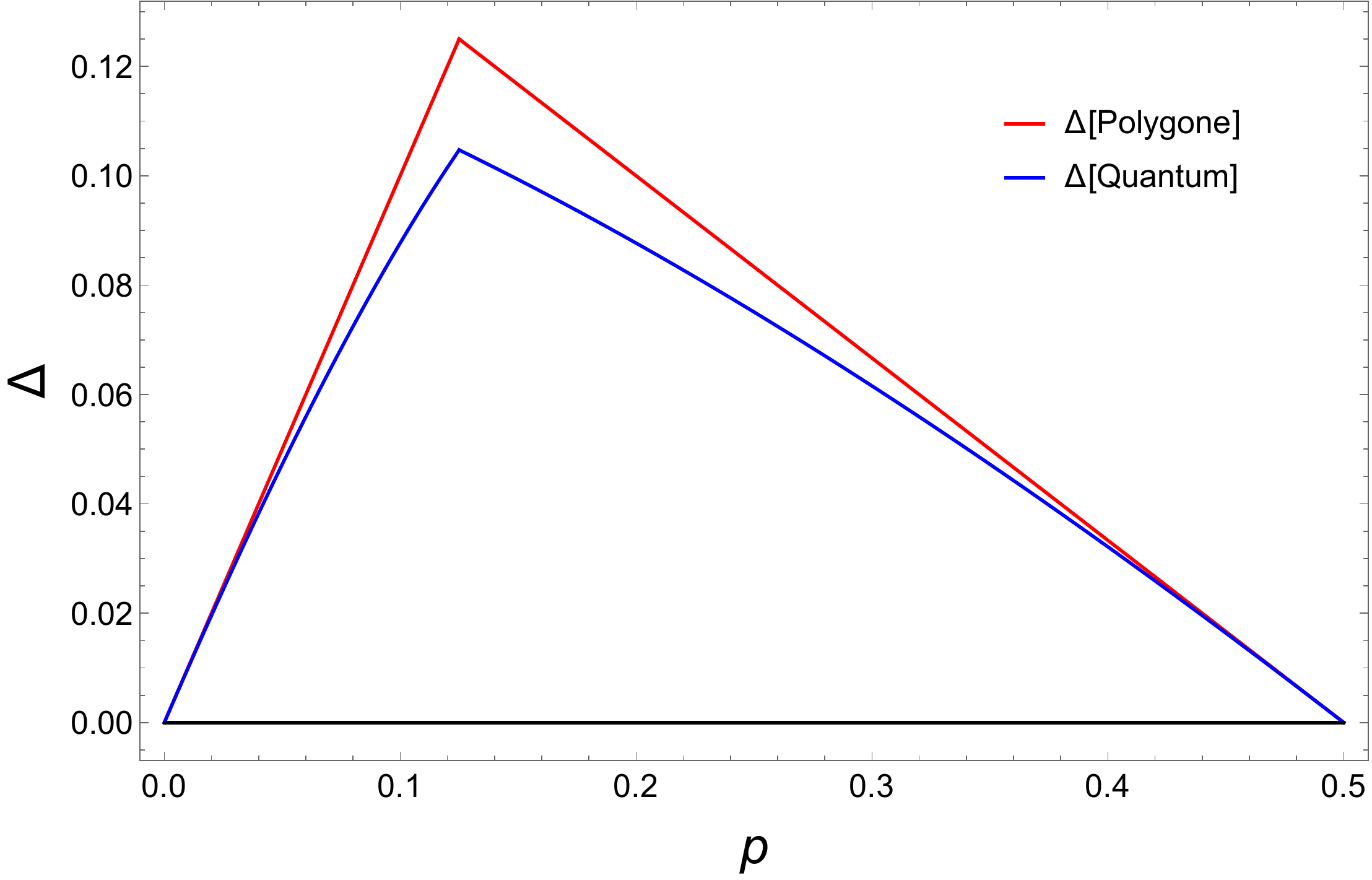}
	\caption{[Color on-line] For $0<p\le\frac{1}{8}$, protocol (a) is advantageous whereas $\frac{1}{8}\le p<\frac{1}{2}$, protocol (b) is advantageous. We have $\Delta[\mbox{Quantum}]<\Delta[\mbox{polygon}]$ for $p\in\left( 0,1/2\right) $.}\label{fig5}
\end{figure}
In the above study we have considered that the set of states $\{\phi_i\}_{i=1}^8$ in Eqs.(\ref{hexas}),(\ref{heptas}),(\ref{pent}), and (\ref{quantum}) are given with uniform probability distribution. This is a true idealistic demand in practical purpose. Here we assume that the states are chosen with a biased probability distribution, we consider ensemble $\{p_i,\phi_i\}_{i=1}^8$ such that $p_i>0~\forall~i~\&~\sum_{i=1}^8p_i=1$. 

Suppose that $p_4=p_4=p$ and rests are $(1-2p)/6$. As discussed earlier, due to the party symmetric nature of the set $\{\phi_i\}_{i=1}^8$, in uniform distribution case any of the party can start the local discrimination protocol. However, for biased scenario the optimal success probability differs depending on the fact which party starts the protocol. In this case we can have two different protocols: (a) Alice starts the protocol, (b) Bob/Charlir starts the protocol.

{\bf Polygon theories:} Straightforward calculations provide, 
\begin{itemize}
	\item[(a)] $P^{succ}_L[\mbox{polygon}]=1-p$ and consequently $\Delta[\mbox{polygon}]=p$,
	\item[(b)] $P^{succ}_L[\mbox{polygon}]=1-\frac{1-2p}{6}$ and consequently $\Delta[\mbox{polygon}]=\frac{1-2p}{6}$.
\end{itemize}
For $0<p\le \frac{1}{8}$ the protocol (a) is advantageous whereas for $\frac{1}{8}\le p<\frac{1}{2}$ the protocol (b) turns out to be superior.

{\bf Quantum theory:}
\begin{itemize}
	\item[(a)] If Alice follows a protocol as of Fig \ref{fig4}, then we have
	$$\Delta[\mbox{Quantum}]=\frac{1}{2}\left( 1-\sqrt{1-4p+8p^2}\right),~\mbox{achieved at}~\theta=\tan^{-1}\left( \frac{2p}{1-2p}\right).$$ 
	\item[(b)] If Bob (or Charlie) follows analogous protocol then we obtain
	$$\Delta[\mbox{Quantum}]=\frac{1}{2}\left( 1-\frac{1}{3}\sqrt{5+4p+8p^2}\right),~\mbox{achieved at}~\theta=\tan^{-1}\left( \frac{1-2p}{2(1-p)}\right).$$
\end{itemize}
Interestingly, for the whole range $0<p<\frac{1}{2}$, we have $\Delta[\mbox{Quantum}]<\Delta[\mbox{polygon}]$ (see Fig \ref{fig5}).

\twocolumngrid


\begin{thebibliography}{99}
\bibitem{self2} $\mathcal{D}(\mathcal{H})$ denotes the collection of all positive semi-definite operator acting on $\mathcal{H}$ and having unit trace. $\mathcal{D}(\mathcal{H})$ is isomorphic to a convex set embedded in the positive cone of $\mathbb{R}^{d^2-1}$, where $d$ be the dimension of $\mathcal{H}$.
\end{thebibliography}
\end{document}